\documentclass[preprint]{elsarticle}
\usepackage{lineno}
\journal{Artificial Intelligence}
\usepackage{mathptmx}

\newcommand\bcmdtab{\noindent\bgroup\tabcolsep=0pt%
  \begin{tabular}{@{}p{10pc}@{}p{20pc}@{}}}
\newcommand\ecmdtab{\end{tabular}\egroup}

\usepackage{mathtools}
\usepackage[ntheorem]{empheq}
\usepackage[thmmarks,amsmath]{ntheorem}
\usepackage{amssymb, amsfonts}
\usepackage{paralist} 
\usepackage{prettyref}
\usepackage[usenames,dvipsnames]{xcolor} 
\usepackage{booktabs} 
\usepackage{enumerate}
\usepackage{graphicx}
\usepackage[utf8]{inputenc}
\usepackage{enumitem}
\usepackage{nicefrac}
\usepackage[hyphens]{url}
\graphicspath{{figures/}}
\usepackage{etaremune}
\usepackage[mathscr]{euscript}
\newproof{proof}{Proof}
\usepackage{xparse}
\usepackage{hyperref}

\newrefformat{cha}{Chapter \ref{#1}}
\newrefformat{sec}{Section \ref{#1}}
\newtheorem{theorem}{Theorem}[section]
\newrefformat{thm}{Theorem \ref{#1}}

\newrefformat{cor}{Corollary \ref{#1}}
\newtheorem{proposition}[theorem]{Proposition}
\newrefformat{prop}{Proposition \ref{#1}}
\newtheorem{remark}{Remark}
\newrefformat{rem}{Remark \ref{#1}}
\newtheorem{lemma}[theorem]{Lemma}
\newrefformat{lem}{Lemma \ref{#1}}
\newrefformat{fig}{Figure \ref{#1}}

\theoremstyle{plain}
\theorembodyfont{\normalfont}
\newtheorem{definition}[theorem]{Definition}
\newrefformat{def}{Definition \ref{#1}}
\newtheorem{example}[theorem]{Example}
\newrefformat{exa}{Example \ref{#1}}

\usepackage{fancyhdr}
\usepackage{fancyvrb}

\newcommand{\mm}{\mathrm}
\newcommand{\mr}{\mathscr}

\newcommand{\mb}[1]{\mathbf{#1}}
\newcommand{\mbb}{\mathbb}
\newcommand{\yref}{\prettyref}

\newcommand{\seq}{\subseteq}
\newcommand{\sneq}{\subsetneq}

\newcommand{\0}{\emptyset}

\newcommand{\la}[1][]{\ensuremath{\xleftarrow{#1}}}

\newcommand{\Iff}{\quad\Longleftrightarrow\quad}
\newcommand{\ioff}{if, and only if, }

\newcommand{\sm}{-}

\newcommand{\naf}{{\mathord{\sim}}}

\newcommand{\cups}{\cup\ldots\cup}

\newcommand{\lr}{_{[\ell,r]}}

\newcommand{\lfp}{{\rm lfp\,}}

\renewcommand{\a}{\alpha}
\renewcommand{\b}{\beta}
\newcommand{\supp}{\mm{supp}\,}
\newcommand{\modG}{\models_\Gamma}
\renewcommand{\IJ}{(\mb I,\mb J)}
\renewcommand{\max}{\mm{max}}
\renewcommand{\P}{\mr P}
\newcommand{\R}{\mr R}
\newcommand{\Mt}{\mm M_{\mb I,t}}
\newcommand{\MMt}{\mm{MM}_{\mb I,t}}
\newcommand{\Tp}[1]{\mm T_{\P,\mb D,#1}}
\newcommand{\Fp}[1]{\Phi_{\P,\mb D,#1}}
\newcommand{\ibt}{_{\mb I,t}}
\newcommand{\mbt}{_{\mb M,t}}
\renewcommand{\t}{^\dagger}
\newcommand{\It}{^{\mb I,t}}
\newcommand{\ITt}{^{\mb I,T,t}}
\newcommand{\HP}{\mm H(\P\It)}
\newcommand{\lra}[1]{\stackrel{\rho_{#1}}{\longleftarrow}}
\renewcommand{\Gamma}{B}

\begin{document}

\begin{frontmatter}

\title{Fixed Point Semantics for Stream Reasoning}

\author{Christian Anti\'c}
\address{Institute of Discrete Mathematics and Geometry\\
         Vienna University of Technology\\
         Wiedner Hauptstra\ss e 8-10, A-1040 Vienna, Austria}
\ead{christian.antic@icloud.com}

\begin{abstract} Reasoning over streams of input data is an essential part of human intelligence. During the last decade {\em stream reasoning} has emerged as a research area within the AI-community with many potential applications. In fact, the increased availability of streaming data via services like Google and Facebook has raised the need for reasoning engines coping with data that changes at high rate. Recently, the rule-based formalism {\em LARS} for non-monotonic stream reasoning under the answer set semantics has been introduced. Syntactically, LARS programs are logic programs with negation incorporating operators for temporal reasoning, most notably {\em window operators} for selecting relevant time points. 
Unfortunately, by preselecting {\em fixed} intervals for the semantic evaluation of programs, the rigid semantics of LARS programs is not flexible enough to {\em constructively} cope with rapidly changing data dependencies. Moreover, we show that defining the answer set semantics of LARS in terms of FLP reducts leads to undesirable circular justifications similar to other ASP extensions. This paper fixes all of the aforementioned shortcomings of LARS. More precisely, we contribute to the foundations of stream reasoning by providing an operational fixed point semantics for a fully flexible variant of LARS and we show that our semantics is sound and constructive in the sense that answer sets are derivable bottom-up and free of circular justifications.
\end{abstract}


\end{frontmatter}

\section{Introduction}

Reasoning over streams of input data is an essential part of human intelligence. During the last decade {\em stream reasoning} has emerged as a research area within the AI-community with many potential applications, e.g., web of things, smart cities, and social media analysis (cf. \cite{Valle09,Mileo17,DellAglio17}). In fact, the increased availability of streaming data via services like Google and Facebook has raised the need for reasoning engines coping with data that changes at high rate.


Logic programs are rule-based systems with the rules and facts being written in a sublanguage of predicate logic extended by a unary operator ``$\naf$'' denoting {\em negation-as-failure} (or {\em default negation}) \citep{Clark78}. While each monotone (i.e., negation-free) logic program has a unique least Herbrand model (with the least model semantics \citep{vanEmden76} being the accepted semantics for this class of programs), for general logic programs a large number of different purely declarative semantics exist. Many of it have been introduced some 20 years ago, among them the {\em answer set semantics} \citep{Gelfond91} and the {\em well-founded semantics} \citep{vanGelder91}. The well-founded semantics, because of its nice computational properties (computing the unique well-founded model is tractable), plays an important role in database theory. However, with the emergence of efficient solvers such as DLV \citep{Leone06}, Smodels \citep{Simons02}, Cmodels \citep{Giunchiglia06}, and Clasp \citep{Gebser12}, programming under answer set semantics led to a predominant declarative problem solving paradigm, called {\em answer set programming} (or {\em ASP}) \citep{Marek99,Lifschitz02}. Answer set programming has a wide range of applications and has been successfully applied to various AI-related subfields such as planning and diagnosis (for a survey see \cite{Brewka11,Eiter09,Baral03}). Driven by this practical needs, a large number of extensions of classical answer set programs have been proposed, e.g. {\em aggregates} (cf. \cite{Faber04,Faber11,Pelov04}), {\em choice rules} \citep{Niemela99}, {\em dl-atoms} \citep{Eiter08a}, and general {\em external atoms} \citep{Eiter05}. For excellent introductions to the field of answer set programming we refer the reader to \cite{Brewka11,Baral03,Eiter09}.

\cite{Beck18} introduced LARS, a Logic-based framework for Analytic Reasoning over Streams, where the semantics of LARS has been defined in terms of FLP-style answer sets \citep{Faber11}. Syntactically, LARS programs are logic programs with negation as failure incorporating operators for temporal reasoning, most notably {\em window operators} for selecting relevant time points. 
Unfortunately, by preselecting {\em fixed} intervals for the semantic evaluation of programs, the rigid semantics of LARS programs is not flexible enough to constructively cope with rapidly changing data dependencies. For example, sentences of the form ``$a$ holds at time point $t$ if $b$ holds at every {\em relevant} time point'' are not expressible within LARS (cf. \yref{exa:Box_b}), as the interval of `relevant' time points changes dynamically, whereas LARS preselects a static interval. Our first step therefore is to refine and simplify \cite{Beck18}'s semantics in \yref{sec:Refined} by employing {\em dynamic} time intervals.

Extensions of the answer set semantics adhere to minimal models or, even more restricting, to models free of unfoundedness. However, FLP-answer sets of stream logic programs may permit undesirable circular justifications similar to other ASP extensions (cf. \cite{Shen14,Antic13}). Fixed point semantics of logic programs (cf. \cite{Fitting02}), on the other hand, are constructive by nature, which suggests to define a fixed point semantics for stream logic programs targeted for foundedness, by recasting suitable operators in such a way that the FLP semantics can be reconstructed or refined, in the sense that a subset of the respective answer sets are selected (sound ``approximation''). The benefit is twofold: by coinciding semantics, we get operable fixed point constructions, and by refined semantics, we obtain a sound approximation that is constructive. For this we recast two well-known fixed point operators from ordinary to stream logic programs, namely the van Emden-Kowalski operator \citep{vanEmden76} and the Fitting operator \citep{Fitting02}. This task turns out to be non-trivial due to the intricate properties of windows \citep{Arasu06,Beck18} and other modal operators occurring in rule heads. We show that the so obtained operators inherit the following characteristic properties: models of a program are characterized by the prefixed points of its associated van Emden-Kowalski operator, and the Fitting operator is monotone with respect to a suitable ordering which guarantees the existence of certain least fixed points, namely the so obtained constructive answer sets. We then show the constructiveness of our fixed point semantics in terms of level mappings \citep{Shen14}. Specifically, we prove that our semantics captures those answer sets which possess a level mapping or, equivalently, which are free of circular justifications, which is regarded as a positive feature.

The rest of the paper is structured as follows. In \yref{sec:LARS} we define the syntax and semantics of stream logic programs first in the vein of \cite{Beck18} (\yref{sec:Beck18}) followed by our refined semantics in \yref{sec:Refined}. \yref{sec:operators} and \ref{sec:semantics} constitute the main part of the paper. More precisely, in \yref{sec:MMt} we define a novel (partial) model operator for the evaluation of rule heads, and in \yref{sec:Tp} and \ref{sec:Fp} we recast the well-known van Emden-Kowalski operator $\mm T_\P$ and the Fitting operator $\Phi_\P$ from ordinary to stream logic programs and prove some non-trivial properties. In \yref{sec:semantics} we then define a fixed point semantics for stream logic programs in terms of the (extended) Fitting operator and prove in our Main \yref{thm:main} the soundness of our approach. Afterwards, in \yref{sec:level} we characterize our semantics in terms of level mappings and conclude that our semantics is sound, constructive, and free of circular justifications.

\section{Stream Logic Programs}\label{sec:LARS}

We denote the set $\mbb N\cup\{\infty\}$ by $\mbb N^\infty$. A {\em partially ordered set} (or {\em poset}) is a pair $\langle L,\leq\rangle$ where $L$ is a set and $\leq$ is a reflexive, antisymmetric, and transitive binary relation on $L$. A {\em lattice} is a poset $\langle L,\leq\rangle$ where every pair of elements $x,y\in L$ has a unique greatest lower bound and least upper bound in $L$. We call $\langle L,\leq\rangle$ {\em complete} if every subset has a greatest lower bound and a least upper bound. For any two elements $x,y\in L$, we define the {\em interval} $[x,y]=\{z\in L\mid x\leq z\leq y\}$. Given a mapping $f:L\to L$, we call $x\in L$ a {\em prefixed point} of $f$ if $f(x)\leq x$, and we call $x$ a {\em fixed point} of $f$ if $f(x)=x$. Moreover, we call $f$ {\em monotone} if $x\leq y$ implies $f(x)\leq f(y)$, for all $x,y\in L$. In case $f$ has a {\em least fixed point}, we denote it by $\lfp f$. Moreover, for a mapping $g:L\times L\to L$ we denote by $g(\,.\,,y)$ the function mapping every $x\in L$ to $g(x,y)\in L$.


\subsection{Streams and Windows}

\begin{quote} In the rest of the paper, $\Sigma$ will denote a finite nonempty set of propositional atoms containing the special symbol $\top$.
\end{quote} A {\em formula} (over $\Sigma$) is defined by the grammar
\begin{align*} \alpha ::= a\mid\neg\alpha\mid\alpha\land\alpha\mid\alpha\lor\alpha\mid\alpha\rightarrow\alpha\mid\Diamond\alpha\mid\Box\alpha\mid @_t\alpha\mid\boxplus\lr\a
\end{align*} where $a\in\Sigma$, $t\geq 1$, and $\ell,r\in\mbb N^\infty$, $\ell\leq r$. We call $\alpha$ (i) {\em $\Box$-free} if it does not contain $\Box$; (ii) {\em monotone} if it does not contain $\neg,\rightarrow,\Box$; and (iii) {\em normal} if it does not contain $\neg,\lor,\rightarrow,\Diamond$.

A {\em stream} (over $\Sigma$) is an infinite sequence $\mb I=I_1I_2\ldots$ of subsets of $\Sigma$, i.e., $I_t\seq\Sigma$ for all {\em time points} $t\geq 1$. We call a stream $\mb J=\mb J_1\mb J_2\ldots$ a {\em substream} of $\mb I$, in symbols $\mb J\seq\mb I$, if $J_t\seq I_t$ for all $t\geq 1$. In the sequel, we omit empty sets in a sequence and write, e.g., $I_1I_3$ instead of $I_1\0 I_3\0\0\ldots$, and we denote the empty sequence $\0\0\ldots$ simply by $\0$. We define the {\em support} of $\mb I$, in symbols $\supp\mb I$, to be the tightest interval $[t_1,t_2]$ containing $\{t\geq 1\mid I_t\neq\0\}$; formally, $t_1=\min\{t\geq 1\mid I_t\neq\0\}$ and $t_2=\max\{t\geq 1\mid I_t\neq\0\}$ in case $\mb I\neq\0$, and $\supp\0=\0$.

A {\em window}{\footnote{\cite{Beck18} employed more sophisticated windows and called them {\em window functions}; for simplicity, we consider here only the windows defined above and note that our results are independent of the particular choice of windows.}} is a function $[\,.\,]$ mapping every stream $\mb I=I_1I_2\ldots$ to the substream $\mb I[\ell,r;t]=I_{\max\{0,t-\ell\}}\ldots I_{t+r}$ of $\mb I$, where $\ell,r\in\mbb N^\infty$, $\ell\leq r$. Note that $[\,.\,]$ and $\supp$ are monotone functions, that is, $\mb I\seq\mb J$ implies $\mb I[\ell,r;t]\seq\mb J[\ell,r;t]$ and $\supp\mb I\seq\supp\mb J$, for all $\ell,r\in\mbb N^\infty$ and $t\geq 1$.

\subsection{Syntax}\label{sec:Syntax}



A ({\em stream logic}) {\em program} $\P$ is a finite nonempty set of {\em rules} of the form
\begin{align}\label{equ:rule} \alpha\la\beta_1,\ldots,\beta_j,\naf\beta_{j+1},\ldots,\naf\beta_k,\quad k\geq j\geq 1,
\end{align} where $\alpha$ is a normal $t$-formula, $\beta_1,\ldots,\beta_k$ are formulas, and $\naf$ denotes {\em negation-as-failure} \citep{Clark78}. We will often write $\la[\rho]$ in expressions of the form \yref{equ:rule} to make the name of the rule explicit. For convenience, we define for a rule $\rho$ of the form \yref{equ:rule}, $\mm H(\rho)=\alpha$, and $\mm B(\rho)=\beta_1\land\ldots\land\beta_j\land\neg\beta_{j+1}\land\ldots\land\neg\beta_k$. As is customary in logic programming, we will interpret every finite set $A$ of formulas as the conjunction $\bigwedge A$ over all formulas in $A$. We call a rule $\rho$ a {\em fact} if $\mm B(\rho)=\top$, and we call $\rho$ {\em ordinary} if $\a,\beta_1,\ldots,\beta_k\in\Sigma$. Moreover, we define $\mm H(\P)$ to be the conjunction of all rule heads occurring in $\P$, that is, $\mm H(\P)=\bigwedge_{\rho\in\P}\mm H(\rho)$.

\subsection{Semantics of  Beck et al. (2018)}\label{sec:Beck18}

We now recall the FLP-style answer set semantics \citep{Faber11} as defined in \cite{Beck18} and we show that their semantics yields counter-intuitive answer sets (cf. \yref{exa:a}).

Let $T$ be a closed interval in $\mbb N$ and let $\Gamma\seq\Sigma$ be a finite set, called the {\em background data}. We define the {\em entailment relation} $\modG$, with respect to $\Gamma$, for all streams $\mb I$, $a\in\Sigma\sm\{\top\}$, formulas $\alpha,\beta$, and all time points $t\in T$:
\begin{enumerate}
\item $\mb I,T,t\modG\top$;
\item $\mb I,T,t\modG a$ if $a\in I_t\cup\Gamma$;
\item $\mb I,T,t\modG\neg\alpha$ if $\mb I,T,t\not\modG\alpha$;
\item $\mb I,T,t\modG\alpha\land\beta$ if $\mb I,T,t\modG\alpha$ and $\mb I,T,t\modG\beta$;
\item $\mb I,T,t\modG\alpha\lor\beta$ if $\mb I,T,t\modG\alpha$ or $\mb I,T,t\modG\beta$;
\item $\mb I,T,t\modG\alpha\rightarrow\beta$ if $\mb I,T,t\not\modG\alpha$ or $\mb I,T,t\modG\beta$;
\item $\mb I,T,t\modG\Diamond\alpha$ if $\mb I,T,t'\modG\alpha$, for some $t'\in T$;
\item $\mb I,T,t\modG\Box\alpha$ if $\mb I,T,t'\modG\alpha$, for all $t'\in T$;
\item $\mb I,T,t\modG @_{t'}\alpha$ if $\mb I,T,t'\modG\alpha$, and $t'\in T$;
\item $\mb I,T,t\modG\boxplus\lr\alpha$ if $\mb I[\ell,r;t],T,t\modG\alpha$.
\end{enumerate} In case $\mb I,T,t\modG\alpha$, we call $\mb I$ a {\em $(t,T)$-model} of $\alpha$.

We wish to evaluate $\P$ with respect to some fixed stream $\mb D$, called the {\em data stream}. We call a stream $\mb I$ an {\em interpretation stream} for $\mb D$ if $\mb D\seq\mb I$, and we say that such an interpretation stream $\mb I$ is a {\em $(t,T)$-model} of $\P$ if $\mb I,T,t\modG\mm B(\rho)\rightarrow\mm H(\rho)$, for all rules $\rho\in\P$. The {\em reduct} of $\P$ with respect to $\mb I$ and $T$ at time point $t$ is given by $$\P\ITt=\{\rho\in\P\mid\mb I,T,t\modG\mm B(\rho)\}.$$


\begin{definition}\label{def:Beck18}[\cite{Beck18}] Let $T$ be a closed interval in $\mbb N$ and let $t\in T$. An interpretation stream $\mb I$ for $\mb D$ is a {\em $(t,T)$-answer stream} of $\P$ (for $\mb D$) if $\mb I$ is a $(t,T)$-model of $\P\ITt$ and there is no $(t,T)$-model $\mb J$ of $\P\ITt$ (for $\mb D$) with $\mb J\sneq\mb I$.
\end{definition}

Note that the minimality condition in \yref{def:Beck18} is given with respect to the same interval $T$, which is crucial. In fact, the following example shows that as a consequence of \yref{def:Beck18}, trivial programs may have infinitely many answer streams which is counter-intuitive from an answer set programming perspective.

\begin{example}\label{exa:a} The ordinary program $P$ consisting of a single fact $a$ has the single answer set $\{a\}$. Given some arbitrary time point $t\geq 1$ for the evaluation of $P$ within the LARS context defined above, we therefore expect $P$ to have the single answer stream $\{a\}_t$. Unfortunately, under \cite{Beck18}'s semantics, $P$ has {\em infinitely} many answer streams: the $(t,[t,t])$-answer stream $\{a\}_t$, the $(t,[t,t+1])$-answer stream $\{a\}_t\0_{t+1}$, the $(t,[t,t+2])$-answer stream $\{a\}_t\0_{t+1}\0_{t+2}$ and so on. 
\end{example}

The reason for the existence of the infinitely many answer streams for the trivial program in \yref{exa:a} is the preselection of the {\em fixed} interval $T$ in \yref{def:Beck18} for the semantic evaluation of programs. As a negative consequence of this choice, which appears to be an artificial simplification of the semantics of programs, is that some specifications which occur in practice cannot be expressed within the LARS language as is demonstrated by the following example.

\begin{example}\label{exa:Box_b} Let $T$ be some interval and let $t\in T$ be some time point. According to \yref{def:Beck18}, the statement ``$a$ holds at $t$ if $b$ holds at every time point in $T$'' is formalized by the single rule $a\la\Box b$ evaluated at time point $t$. Now consider the slighly different statement ``$a$ holds at $t$ if $b$ holds at every relevant time point in the support\footnote{Recall from \yref{sec:Syntax} that the support of a stream is given by the tightest intervall containing all non-empty (i.e., relevant) time points.} of the input data.'' As natural as this statement seems, it is {\em not} expressible within the LARS language. The intuitive reason is that the support function is flexible and depends on the data, whereas the preselected interval $T$ is fixed by the programmer and therefore does not depend on the data. In a sense, preselecting fixed intervals for the semantic evaluation of programs contradicts the very idea of stream reasoning which aims at coping with data that changes at a high rate by incorporating window operators on a {\em syntactic} level for selecting relevant time points. Arguably, it is therefore more natural to formalize the first statement by the rule $a\la\boxplus_T\Box b$ thus syntactically encoding the restricted interval $T$, and interpreting $a\la\Box b$ as a formalization of the second statement (cf. \yref{exa:Box_b_2}).
\end{example}

\subsection{Refined Semantics}\label{sec:Refined}

We refine the FLP-style semantics of \cite{Beck18} (cf. \yref{def:Beck18}) by employing {\em dynamic} intervals. For this we first refine the entailment relation by using the support function in the definition of $\Box$ and $\Diamond$ for dynamically computing intervals instead of the fixed interval $T$ used by \cite{Beck18}:

\begin{enumerate}
\item $\mb I,t\modG\Diamond\alpha$ if $\mb I,t'\modG\alpha$, for some $t'\in\supp\mb I$;
\item $\mb I,t\modG\Box\alpha$ if $\mb I,t'\modG\alpha$, for all $t'\in\supp\mb I$;
\item $\mb I,t\modG @_{t'}\alpha$ if $\mb I,t'\modG\alpha$, for $t'\geq 1$.
\end{enumerate} In case $\mb I,t\modG\alpha$, we call $\mb I$ a {\em $t$-model} of $\alpha$, and we call $\alpha$ {\em $t$-consistent} (resp., {\em $t$-inconsistent}) if $\alpha$ has at least one (resp., no) $t$-model. For convenience, we call $\alpha$ a {\em $t$-formula} if $\alpha$ is $t$-consistent.

\begin{example} The formula $\boxplus_{[0,0]}@_2 a$ is 1-inconsistent since $@_2$ is a reference to time point 2 which is outside the scope of the window $\boxplus_{[0,0]}$ evaluated at time point 1. More precisely, let $\mb I=I_1I_2\ldots$ be an arbitrary stream and compute $\mb I[0,0;1]=I_1$ which implies $I_1,2\not\modG a$---so $\mb I$ is not a 1-model of $\alpha$.
\end{example}

\begin{remark}\label{rem:incon} Note that $t$-inconsistency of normal formulas can be easily verified by a syntactic check as in the example above and in the rest of the paper we assume that (normal) formulas occurring in rule heads are $t$-consistent, for all relevant $t$.
\end{remark}

We can now refine and simplify the definition of answer streams by omitting the reference to interval $T$ which gives a more natural minimality condition.

\begin{definition}\label{def:answer} An interpretation stream $\mb I$ for $\mb D$ is a {\em $t$-answer stream} of $\P$ (for $\mb D$) if $\mb I$ is a substream minimal $t$-model of $\P\It$.
\end{definition}

\begin{example} The ordinary program $P$ of \yref{exa:a} consisting of the single fact $a$ has the single $t$-answer stream $\{a\}_t$ as expected.
\end{example}

\begin{example}\label{exa:Box_b_2} The two statements in \yref{exa:Box_b} are formalized according to \yref{def:answer} by the two rules $a\la\boxplus_T\Box b$ and $a\la\Box b$, respectively, as desired.
\end{example}

We now illustrate the above concepts in more detail with the following running example.

\begin{example}\label{exa:running} Consider the program $\P$ consisting of the following rules:
\begin{align*} 
@_2 a&\lra{1}\naf @_7 c & \boxplus_{[1,\infty]}\Box c&\lra{3}\naf @_2 a\\
\boxplus_{[\infty,0]}\Box a&\lra{2}\naf c & \boxplus_{[2,3]}\Box(a\land b)&\lra{4}\boxplus_{[0,1]}\Diamond c,\Box d.
\end{align*} Let the background data $\Gamma$ consist of the single proposition $d$, and let the data stream $\mb D$ be given by
\begin{align*} \mb D=\{a\}_1\{a,b\}_5\{c\}_{10}.
\end{align*} That is, the propositions $a$ and $b$ hold at time point 5 and so on. Then, the 5-answer streams of $\P$ (for $\mb D$) are given by:
\begin{align*} 
\mb I &= \{a\}_1\{a,b\}_3\{a,b,c\}_4\{a,b,c\}_5\{a,b,c\}_6\{a,b,c\}_7\{a,b,c\}_8\{c\}_9\{c\}_{10};\\
\mb J &= \{a\}_1\{a\}_2\{a\}_3\{a\}_4\{a,b\}_5\{c\}_{10}.
\end{align*} For instance, we verify that $\mb I$ is indeed a 5-answer stream of $\P$. First of all, note that $\mb I$ is a 5-model of $\P$ (for $\mb D$): (i) as $c$ holds in $\mb I$ at time points 5 and 7, we have $\mb I,5\not\modG\naf @_7 c$ and $\mb I,5\not\modG\naf c$ which implies $\mb I,5\modG\rho_1$ and $\mb I,5\modG\rho_2$; (ii) as $c$ holds at every time point in the interval $[4,10]$, we have $\mb I,5\modG\boxplus_{[1,\infty]}\Box c$ which implies $\mb I,5\modG\rho_3$; and (iii) as $a$ and $b$ hold at every time point in $[3,8]$, we have $\mb I,5\modG\boxplus_{[2,3]}\Box(a\land b)$ which implies $\mb I,5\modG\rho_4$.

Now we argue that $\mb I$ is a {\em minimal} 5-model of $\P^{\mb I,5}=\{\rho_3,\rho_4\}$. To this end, suppose $\mb I'=\mb I'_1\mb I'_2\ldots$ is a 5-model of $\P^{\mb I,5}$ for $\mb D$ with $\mb D\seq\mb I'\seq\mb I$. Then, since $\mb I',5\modG\mm B(\rho_3)$ and $\mb I',5\modG\mm B(\rho_4)$, for $\mb I'$ to be a 5-model of $\P$ we must have $\mb I',5\modG\mm H(\rho_3)$ and $\mb I',5\modG\mm H(\rho_4)$; but this is equivalent to $\mb I'[1,\infty;5],t\modG c$, for all $t\in [4,10]$, and $\mb I'[2,3;5],t'\modG a\land b$, for all $t'\in [3,8]$ where $\mb I'[1,\infty;5]=I'_4\ldots I'_{10}$ and $\mb I'[2,3;5]= I'_3\ldots I'_8$, respectively. That is, we have $c\in I'_t$, for all $t\in [4,10]$, and $a,b\in I'_{t'}$, for all $t'\in [3,8]$---but this, together with $\mb D\seq\mb I'\seq\mb I$, immediately implies $\mb I'=\mb I$ which shows that $\mb I$ is indeed a minimal 5-model of $\P^{\mb I,5}$ and therefore a 5-answer stream of $\P$.
\end{example}

It is important to emphasize that we can capture \cite{Beck18}'s semantics as follows.

\begin{proposition}\label{prop:boxplus_P} Let $T=[t_1,t_2]$ be an interval and let $t\in T$ be some time point. An interpretation stream $\mb I$ for $\mb D$ is a $(t,T)$-answer stream of $P$ \ioff $\mb I\cup\{\#\}_{t_1}\ldots\{\#\}_{t_2}$ is a $t$-answer stream of $$\boxplus_T P\cup\{@_t\#\mid t\in T\},$$ where $\#$ is a special symbol not occurring in $\Sigma$ and $\boxplus_T P$ consists of all rules of the form
\begin{align*} \boxplus_T\rho=\boxplus_T\alpha\la\boxplus_T\beta_1,\ldots,\boxplus_T\beta_j,\naf\boxplus_T\beta_{j+1},\ldots,\naf\boxplus_T\beta_k,\quad\rho\in P.
\end{align*}
\end{proposition}

At this point, we have successfully extended the FLP-style answer set semantics from ordinary to stream logic programs by refining \cite{Beck18}'s semantics. Unfortunately, as for other program extensions (cf. \cite{Shen14,Antic13}), our FLP-style semantics may permit circular justifications as is demonstrated by the following example.

\begin{example}\label{exa:R} Consider the program $\R$ consisting of the following two rules:
\begin{align*} 
a&\lra{1}\Box b\\
b&\lra{2}\Box a.
\end{align*} We argue that the $t$-model $\{a,b\}_t$ of $\R$ is a $t$-answer stream of $\R$, for every $t\geq 1$ (and $\mb D=\Gamma=\0$): (i) The empty stream $\0$ is not a $t$-model of $\R^{\{a,b\}_t,t}=\R$ since both rules fire in $\0$; (ii) the stream $\{a\}_t$ is not a $t$-model of $\R$ since $\rho_2$ fires; (iii) the stream $\{b\}_t$ is not a $t$-model of $\R$ since $\rho_1$ fires. This shows that $\{a,b\}_t$ is indeed a subset minimal $t$-model of $\R^{\{a,b\}_t,t}$ and, hence, a $t$-answer stream of $\R$.
\end{example}


In the next two sections, we will develop the tools for formalizing the reasoning in \yref{exa:running} in an operational setting (cf. \yref{exa:Fp}) while avoiding circular justifications.

\section{Fixed Point Operators}\label{sec:operators}

In this section, we recast the following well-known fixed point operators from ordinary to stream logic programs: (i) the van Emden-Kowalski operator $\mm T_\P$ \citep{vanEmden76}, and (ii) the Fitting operator $\Phi_\P$ \citep{Fitting02}. This task turns out to be non-trivial due to the intricate properties of windows and other modal operators occurring in rule heads.

\begin{quote} In the rest of the paper, let $\mb I$ be a stream, let $\mb D$ be some data stream, let $\Gamma$ be some background data, and let $t\geq 1$ be some fixed time point.
\end{quote}

\subsection{The Model Operator}\label{sec:MMt}

In this subsection, we define an operator for the evaluation of rule heads. Specifically, given a normal $t$-formula $\alpha$, we wish to construct a $t$-model of $\alpha$ which is in some sense minimal with respect to a given stream $\mb I$ (cf. \yref{thm:MMt}).

\begin{definition}\label{def:MMt} For normal $t$-formulas $\alpha$ and $\beta$, and for $a\in\Sigma$, we define the {\em partial model operator} $\Mt$ at time point $t$ and with respect to $\mb I$, inductively as follows:
\begin{align*} 
&\Mt(a)=
  \begin{cases} 
  \{a\}_t & \text{if }a\not\in\Gamma,\\
  \0 & \text{if }a\in\Gamma;
  \end{cases}\\
&\Mt(\alpha\land\beta)=\Mt(\alpha)\cup\Mt(\beta);\\
&\Mt(\Box\alpha)=\bigcup_{t'\in\supp\mb I}\mm M_{\mb I,t'}(\alpha);\\
&\Mt(@_{t'}\alpha)=\mm M_{\mb I,t'}(\alpha);\\
&\Mt(\boxplus\lr\alpha)=\mm M_{\mb I[\ell,r;t],t}(\alpha).
\end{align*} Finally, define the {\em model operator} $\MMt$ to be the twofold application of $\Mt$, that is,
\begin{align*} \MMt(\alpha)=\mm M_{\Mt(\alpha),t}(\alpha).
\end{align*}
\end{definition}

One can easily derive the following computation rules for the model operator:
\begin{align*} 
&\MMt(a) = \Mt(a);\\
&\MMt(@_{t'}\alpha) = \mm{MM}_{\mb I,t'}(\alpha);\\
&\MMt(\boxplus\lr\alpha) = \mm{MM}_{\mb I[\ell,r;t],t}(\alpha).
\end{align*}

In case $\alpha$ is $\Box$-free, we will often write $\mm M_t(\alpha)$ instead of $\Mt(\alpha)$ to indicate that the evaluation of $\Mt$ does not depend on $\mb I$.

\begin{example}\label{exa:MMt} Let $\alpha$ be the $\Box$-free normal 1-formula $\boxplus_{[0,0]} @_1a\land @_2b$, and compute
\begin{align*} \mm M_{\0,1}(\boxplus_{[0,0]}@_1a\land @_2b)=\mm M_{\0[0,0;1],1}(a)\cup\mm M_{\0,2}(b)=\{a\}_1\{b\}_2
\end{align*} which is a 1-model of $\alpha$. On the other hand, for the normal 1-formula $\beta=\Box a\land b$ containing $\Box$, we obtain
\begin{align*} \mm M_{\0,1}(\Box a\land b)=\mm M_{\0,1}(\Box a)\cup\mm M_{\0,1}(b)=\mm M_{\0,1}(b)=\{b\}_1
\end{align*} which is {\em not} a 1-model of $\beta$; however, by applying $\mm M_{\0,1}$ twice, we do obtain a 1-model of $\beta$:
\begin{align*} \mm{MM}_{\0,1}(\Box a\land b)=\mm M_{\{b\}_1,1}(\Box a\land b)=\mm M_{\{b\}_1,1}(a)\cup\mm M_{\{b\}_1,1}(b)=\{a,b\}_1.
\end{align*} Intuitively, to obtain a 1-model of $\beta$, we have to apply $\mm M_{\0,1}$ twice as the subformula $b$ induces an expansion of the support of the generated stream which has to be taken into account for the generation of a 1-model for $\Box a$ (note that conjunctions are treated separately by the partial model operator).
\end{example}

\yref{exa:MMt} indicates that $\Box$ requires a special treatment. In fact, if $\alpha$ is $\Box$-free then $\Mt$ does not depend on $\mb I$ and, consequently, in this case $\Mt(\alpha)$ and $\MMt(\alpha)$ coincide which simplifies the matters significantly.

\begin{proposition}\label{prop:B-free} For every $\Box$-free normal $t$-formula $\alpha$, $\Mt(\alpha)=\mm M_{\mb J,t}(\alpha)$ holds for all streams $\mb I$ and $\mb J$; consequently, $\Mt(\alpha)=\MMt(\alpha)$.
\end{proposition}
\begin{proof} By definition of the partial model operator, $\Mt(\alpha)$ depends on $\mb I$ only if $\alpha$ contains $\Box$. The second assertion follows from the first with $\mb J=\Mt(\alpha)$ and $\MMt(\alpha)=\mm M_{\mb J,t}(\alpha)$.
\end{proof}

In the next two propositions, we show some monotonicity properties of the (partial) model operator.

\begin{proposition}\label{prop:MMt-mon} For every normal $t$-formula $\alpha$ and all streams $\mb K$ and $\mb I$, $\mb I\seq\mb J$ implies $\Mt(\alpha)\seq\mm M_{\mb J,t}(\alpha)$ and $\MMt(\alpha)\seq\mm{MM}_{\mb J,t}(\alpha)$. Moreover, $A\seq B$ implies $\Mt(A)\seq\Mt(B)$ and $\MMt(A)\seq\MMt(B)$, for all finite sets $A$ and $B$ of normal $t$-formulas.
\end{proposition}
\begin{proof} The first inclusion can be shown by a straightforward structural induction on $\alpha$, so we prove here only the case $\alpha=\Box\beta$, for some normal $t$-formula $\beta$:
\begin{align*} 
\Mt(\Box\beta)
  &= \bigcup_{t'\in\supp\mb I}\mm M_{\mb I,t'}(\beta)
  &\seq\bigcup_{t'\in\supp\mb J}\mm M_{\mb I,t'}(\beta)
  &\stackrel{\mm I\mm H}\seq\bigcup_{t'\in\supp\mb J}\mm M_{\mb J,t'}(\beta)
  &=\mm M_{\mb J,t}(\Box\beta).
\end{align*} The second inclusion, $\MMt(\alpha)\seq\mm{MM}_{\mb J,t}(\alpha)$, is a direct consequence of the first. Finally, the second part of the proposition is an immediate consequence of the first part and the definition of the (partial) model operator.
\end{proof}

\begin{proposition}\label{prop:seq} For every normal $t$-formula $\alpha$, $\supp\Mt(\alpha)=\supp\MMt(\alpha)$ and $$\Mt(\alpha)\seq\MMt(\alpha).$$
\end{proposition}
\begin{proof} The first identity can be proved by a straightforward structural induction on $\alpha$.

We prove the inclusion by structural induction on $\alpha$. The only non-trivial case is $\alpha=\Box\beta$, for some normal $t$-formula $\beta$:
\begin{align*} 
  \Mt(\Box\beta)
    &=\bigcup_{t'\in\supp\mb I}\mm M_{\mb I,t'}(\beta)\\
    &\stackrel{\mm I\mm H}\seq\bigcup_{t'\in\supp\mb I}\mm{MM}_{\mb I,t'}(\beta)\\
    &=\bigcup_{t'\in\supp\mb I}\mm M_{\mm M_{\mb I,t'}(\beta),t'}(\beta)\\
    &\seq\bigcup_{t'\in\supp\mb I}\mm M_{\mm M_{\mb I,t}(\Box\beta),t'}(\beta)\\
    &\seq\bigcup_{t'\in\supp\mb I}\mm M_{\mm M_{\mb I,t}(\Box\beta),t'}(\Box\beta)\\
    &=\bigcup_{t'\in\supp\mb I}\mm{MM}_{\mb I,t'}(\Box\beta)\\
    &=\MMt(\Box\beta)
\end{align*} where the second and third inclusion follows from \yref{prop:MMt-mon} together with
\begin{align*} \mm M_{\mb I,t'}(\beta)\seq\Mt(\Box\beta)\quad\text{for all }t'\in\supp\mb I,
\end{align*} and the last equality holds since:
\begin{align}\label{equ:box} \MMt(\Box\beta)=\mm{MM}_{\mb I,t'}(\Box\beta)\quad\text{for all }t'\in\supp\mb I.
\end{align} To prove \yref{equ:box}, first note that, by definition, $\Mt(\Box\beta)=\mm M_{\mb I,t'}(\Box\beta)$ holds for all $t'\in\supp\mb I$; consequently:
\begin{align*} 
  \MMt(\Box\beta)
    &=\mm M_{\mm M_{\mb I,t}(\Box\beta),t}(\Box\beta)\\
    &=\bigcup_{t''\in\supp\Mt(\Box\beta)}\mm M_{\Mt(\Box\beta),t''}(\beta)\\
    &=\bigcup_{t''\in\supp\mm M_{\mb I,t'}(\Box\beta)}\mm M_{\mm M_{\mb I,t'}(\Box\beta),t''}(\beta)\\
    &=\mm M_{\mm M_{\mb I,t'}(\Box\beta),t'}(\Box\beta)\\
    &=\mm{MM}_{\mb I,t'}(\Box\beta).
\end{align*}
\end{proof}

It will often be convenient to separate a proof into a $\Box$-free and a general case. Therefore we define the translation $\alpha\ibt$ of $\alpha$ with respect to $\mb I$ at time point $t$ to be the homomorphic extension to all normal $t$-formulas of:
\begin{align*}
& (\Box\alpha)\ibt = \bigwedge_{t'\in\supp\mb I}@_{t'}\a_{\mb I,t'};\\
& (\boxplus\lr\alpha)\ibt = \boxplus\lr\alpha_{\mb I[\ell,r;t],t},
\end{align*} where we interpret the empty conjunction in $(\Box\alpha)_{\0,t}$ as $\top$. Intuitively, $.\ibt$ eliminates every $\Box$ occurring in $\alpha$ while preserving the meaning of $\alpha$ in the following sense.

\begin{proposition}\label{prop:ibt} For every normal $t$-formula $\alpha$, and for all streams $\mb I$ and $\mb J$ with $\supp\mb I=\supp\mb J$, we have $\mb I,t\modG\alpha$ if, and only if, $\mb I,t\modG\alpha_{\mb J,t}$.
\end{proposition}
\begin{proof} A straightforward structural induction on $\alpha$. 
\end{proof}

Interestingly, the next proposition shows that we can simulate the model operator by the partial model operator applied to an appropriate translation of the input formula.

\begin{proposition}\label{prop:mbt} For every normal $t$-formula $\alpha$, $\Mt(\alpha)=\Mt(\alpha\ibt)$ and, consequently, $\MMt(\alpha)=\Mt(\alpha\mbt)$ with $\mb M=\Mt(\alpha)$.
\end{proposition}
\begin{proof} The first identity can be proved by a straightforward structural induction on $\alpha$, and the second identity follows from the first, i.e., $\MMt(\alpha)=\mm M_{\mb M,t}(\alpha)=\Mt(\alpha\mbt)$.
\end{proof}

Monotone formulas inherit their name from the following property.

\begin{proposition}\label{prop:mon} For every monotone formula $\alpha$, $\mb I,t\modG\alpha$ implies $\mb J,t\modG\alpha$, for all streams $\mb I\seq\mb J$.
\end{proposition}

We are now ready to prove our first theorem which shows that $\MMt(\alpha)$ is a $t$-model of $\alpha$ which is in some sense ``minimal'' (with respect to $\mb I$).

\begin{theorem}\label{thm:MMt} For every normal $t$-formula $\alpha$, $\MMt(\alpha)$ is a $t$-model of $\alpha$, that is,
\begin{align*} \MMt(\alpha),t\modG\alpha.
\end{align*} In case $\alpha$ is $\Box$-free, a single application of $\Mt$ suffices, that is, $\Mt(\alpha),t\modG\alpha$. Moreover, if $\mb I$ is a $t$-model of $\alpha$, then $\MMt(\alpha)\seq\mb I$.
\end{theorem}
\begin{proof} We start with the second assertion and prove by structural induction on $\alpha$ that in case $\alpha$ is $\Box$-free, $\Mt(\alpha),t\modG\alpha$. The induction hypothesis  $\alpha=a\in\Sigma$, and the case $\alpha=@_{t'}\beta$ are straightforward. In what follows, let $\beta$ and $\gamma$ denote normal $t$-formulas. For $\alpha=\b\land\gamma$, we have $\Mt(\b\land\gamma)=\Mt(\beta)\cup\Mt(\gamma)$ and, by induction hypothesis, $\Mt(\beta),t\modG\beta$ and $\Mt(\gamma),t\modG\gamma$. Since $\beta$ and $\gamma$ are $\Box$-free, we have $\Mt(\beta)\cup\Mt(\gamma),t\modG\b\land\gamma$ as a consequence of \yref{prop:mon} (recall that $\Box$-free normal formulas are monotone). Finally, for $\alpha=\boxplus\lr\beta$ we have $\Mt(\boxplus\lr\beta)=\Mt(\beta)$ and, by induction hypothesis, $\Mt(\beta),t\modG\beta$. Since $\boxplus\lr\beta$ is $t$-consistent by assumption, $\Mt(\beta)=\Mt(\beta)[\ell,r;t]$ (cf. \yref{rem:incon} and \ref{rem:MMt}) and, hence, $\Mt(\beta)[\ell,r;t],t\modG\beta$ which is equivalent to $\Mt(\boxplus\lr\beta),t\modG\boxplus\lr\beta$. 

We now turn to the general case and prove that $\MMt(\alpha),t\modG\alpha$ holds for any normal $t$-formula $\alpha$, by first translating $\alpha$ into a $\Box$-free formula, and then referring to the first part of the proof. Let $\mb M=\Mt(\alpha)$. Since $\alpha\mbt$ is $\Box$-free, we know from the first part of the proof that
\begin{align}\label{equ:MMt-1} \Mt(\alpha\mbt),t\modG\alpha\mbt.
\end{align} By \yref{prop:mbt},
\begin{align}\label{equ:MMt-2} \MMt(\alpha)=\Mt(\alpha\mbt).
\end{align} From \yref{equ:MMt-1} and \yref{equ:MMt-2} we infer
\begin{align*} \MMt(\alpha),t\modG\alpha\mbt.
\end{align*} Now since $\supp\mb M=\supp\MMt(\alpha)$ (cf. \yref{prop:seq}), \yref{prop:ibt} proves our claim.

Finally, a straightforward structural induction on $\alpha$ shows $\MMt(\alpha)\seq\mb I$.
\end{proof}

\begin{remark}\label{rem:MMt} We want to emphasize that the requirement in \yref{thm:MMt} of $\alpha$ being $t$-consistent is essential. For instance, reconsider the 1-inconsistent normal formula $\alpha=\boxplus_{[0,0]}@_2a$ of \yref{rem:incon}, and compute $\mm{MM}_{\mb I,1}(\alpha)=\{a\}_2$ which is {\em not} a 1-model of $\alpha$.
\end{remark}



\subsection{The van Emden-Kowalski Operator}\label{sec:Tp}

We are now ready to extend the well-known van Emden-Kowalski operator to the class of stream logic programs. 

\begin{definition} We define the {\em van Emden-Kowalski operator} $\Tp t$ of $\P$ (for $\mb D$ at time point $t$), for every stream $\mb I$, by
\begin{align*} \Tp t(\mb I)=\mb D\cup\MMt(\{\mm H(\rho)\mid\rho\in\P:\mb I,t\modG\mm B(\rho)\}).
\end{align*} 
\end{definition}

As for ordinary logic programs \citep{vanEmden76}, prefixed points of the van Emden-Kowalski operator $\Tp t$ characterize the models of $\P$ (for $\mb D$ at time point $t$).

\begin{theorem}\label{thm:Tp} A stream $\mb I$ is a $t$-model of $\P$ if, and only if, $\mb I$ is a prefixed point of $\Tp t$.
\end{theorem}
\begin{proof} Suppose $\mb I$ is a $t$-model of $\P$ and note that this is equivalent to
\begin{align}\label{equ:Tp-1} \mb I,t\modG\HP.
\end{align} Moreover, note that we can rewrite the van Emden-Kowalski operator more compactly as
\begin{align*} \Tp t(\mb I)=\mb D\cup\MMt(\HP).
\end{align*} So we have to show $\MMt(\HP)\seq\mb I$ (recall that $\mb D\seq\mb I$ holds by assumption), but this follows directly from \yref{equ:Tp-1} together with the last part of \yref{thm:MMt}.

For the other direction, we show that $\MMt(\HP)\seq\mb I$ implies $\mb I,t\modG\HP$. Let $\mb M=\Mt(\HP)$. By \yref{prop:seq},
\begin{align}\label{equ:Tp-2} \mb M=\Mt(\HP)\seq\MMt(\HP)=\mm M_{\mb M,t}(\HP)\seq\mb I.
\end{align} On the other hand, $\mb M\seq\mb I$ and \yref{prop:MMt-mon} imply
\begin{align}\label{equ:Tp-3} \mm M_{\mb M,t}(\HP)\seq\Mt(\HP).
\end{align} Consequently, from \yref{equ:Tp-2} and \yref{equ:Tp-3} we infer
\begin{align}\label{equ:Tp-4} \Mt(\HP)=\MMt(\HP).
\end{align} Intuitively, \yref{equ:Tp-4} means that in case $\MMt(\HP)\seq\mb I$, one application of $\Mt$ suffices (cf. \yref{thm:MMt} and \yref{exa:MMt}). Moreover, \yref{prop:mbt} implies
\begin{align}\label{equ:Tp-5} \Mt(\HP)=\Mt(\HP\ibt)\seq\mb I.
\end{align} Since $\Mt(\HP\ibt)$ is a $t$-model of $\HP\ibt$ (cf. \yref{thm:MMt}), $\Mt(\HP\ibt)\seq\mb I$ holds by \yref{equ:Tp-5}, and $\HP\ibt$ is monotone, \yref{prop:mon} and \yref{equ:Tp-5} imply
\begin{align}\label{equ:Tp-6} \mb I,t\modG\HP\ibt.
\end{align} Finally, \yref{prop:ibt} and \yref{equ:Tp-6} imply $\mb I,t\modG\HP$.
\end{proof}

\begin{example}\label{exa:Tp} Reconsider the program $\P$ of \yref{exa:running} consisting of the following rules:
\begin{align*} 
@_2 a&\lra{1}\naf @_7 c       & \boxplus_{[1,\infty]}\Box c&\lra{3}\naf @_2 a\\
\boxplus_{[\infty,0]}\Box a&\lra{2}\naf c  & \boxplus_{[2,3]}\Box(a\land b)&\lra{4}\boxplus_{[0,1]}\Diamond c,\Box d.
\end{align*} We have argued in \yref{exa:running} that the interpretation stream $$\mb I=\{a\}_1\{a,b\}_3\{a,b,c\}_4\{a,b,c\}_5\{a,b,c\}_6\{a,b,c\}_7\{a,b,c\}_8\{c\}_9\{c\}_{10}$$ of $\P$ for $\mb D=\{a\}_1\{a,b\}_5\{c\}_{10}$ and $\Gamma=\{d\}$ is a 5-model of $\P$. Now we want to rigorously prove that $\mb I$ is a 5-model of $\P$ by showing that $\mb I$ is a prefixed point of $\Tp 5$. We compute:
\begin{align*} 
\mb M &= \mm M_{\mb I,5}(\{\mm H(\rho_3),\mm H(\rho_4)\})\\
  &= \mm M_{\mb I,5}(\boxplus_{[1,\infty]}\Box c\land\boxplus_{[2,3]}\Box(a\land b))\\
  &= \mm M_{\mb I,5}(\boxplus_{[1,\infty]}\Box c)\cup\mm M_{\mb I,5}(\boxplus_{[2,3]}\Box(a\land b))\\
  &= \bigcup_{t\in\supp\mb I[1,\infty;5]}\mm M_{\mb I[1,\infty;5],t}(c)\cup\bigcup_{t\in\supp\mb I[2,3;5]}\mm M_{\mb I[2,3;5],t}(a\land b)\\
  &= \{a,b\}_3\{a,b,c\}_4\{a,b,c\}_5\{a,b,c\}_6\{a,b,c\}_7\{a,b,c\}_8\{c\}_9\{c\}_{10}
\end{align*} and
\begin{align*} 
\Tp 5(\mb I)
  &=\mb D\cup\mm{MM}_{\mb I,5}(\{\mm H(\rho_3),\mm H(\rho_4)\})\\
  &=\mb D\cup\mm M_{\mb M,5}(\boxplus_{[1,\infty]}\Box c)\cup\mm M_{\mb M,5}(\boxplus_{[2,3]}\Box(a\land b))\\
  &=\mb D\cup\bigcup_{t\in\supp\mb M[1,\infty;5]}\mm M_{\mb M[1,\infty;5],t}(c)\cup\bigcup_{t\in\supp\mb M[2,3;5]}\mm M_{\mb M[2,3;5],t}(a\land b)\\
  &=\mb D\cup\mb M\\
  &=\{a\}_1\{a,b\}_3\{a,b,c\}_4\{a,b,c\}_5\{a,b,c\}_6\{a,b,c\}_7\{a,b,c\}_8\{c\}_9\{c\}_{10}\\
  &=\mb I.
\end{align*}
\end{example}



\subsection{The Fitting Operator}\label{sec:Fp}

In the presence of negation, the van Emden-Kowalski operator is non-monotonic and cannot be iterated bottom-up. We therefore extend the (3-valued) Fitting operator \cite{Fitting02} from ordinary to stream logic programs as follows. Firstly, we define a {\em 3-valued stream} to be a pair of streams $\IJ$ with $\mb I\seq\mb J$ or, equivalently, a sequence of pairs $(I_1,J_1)(I_2,J_2)\ldots$ with $I_t\seq J_t$ for all $t\geq 1$, with the intuitive meaning that every $a\in I_t$ (resp., $a\not\in J_t$) is {\em true} (resp., {\em false}) at time point $t$, whereas every $a\in J_t\sm I_t$ is {\em undefined} at $t$.

We then define the {\em precision ordering}{\footnote{The precision ordering corresponds to the {\em knowledge ordering} $\leq_k$ in \cite{Fitting02}; cf. \cite{Denecker04}.}} $\seq_p$ on the set of all 3-valued streams by
\begin{align*} (\mb I,\mb J)\seq_p (\mb I',\mb J')\Iff\mb I\seq\mb I'\text{ and }\mb J'\seq\mb J.
\end{align*} Intuitively, $\IJ\seq_p (\mb I',\mb J')$ means that $(\mb I',\mb J')$ is a ``tighter'' interval inside $\IJ$. The maximal elements with respect to $\seq_p$ are exactly the (2-valued) streams where we identify each stream $\mb I$ with $(\mb I,\mb I)$. Note that since distinct streams have no upper bound with respect to the precision ordering, the set of all 3-valued streams is not a lattice.

We extend the entailment relation to 3-valued streams as follows.

\begin{definition}\label{def:3-val} For every 3-valued stream $\IJ$ and formula $\alpha$,
\begin{align*} \IJ,t\modG\alpha\Iff\mb K,t\modG\alpha\text{ for every }\mb K\in [\mb I,\mb J].
\end{align*} 
\end{definition}

The intuition behind \yref{def:3-val} is as follows. Recall that a formula $\alpha$ containing $\neg,\rightarrow,$ or $\Box$ may be non-monotone in the sense of \yref{prop:mon}, and in this case we have to take all possible extensions $\mb K\in [\mb I,\mb J]$ of $\mb I$ into account.


Now define the {\em Fitting operator} $\Fp t$ of $\P$ for $\mb D$ at time point $t$, for every 3-valued stream $\IJ$, by
\begin{align*} \Fp t\IJ=\mb D\cup\MMt(\{\mm H(\rho)\mid\rho\in\P:\IJ,t\modG\mm B(\rho)\}).
\end{align*} 

The only difference between $\Fp t$ and $\Tp t$ is that $\Fp t$ evaluates the body of a rule in a 3-valued stream, which guarantees the monotonicity of $\Fp t$ with respect to the precision ordering (cf. \yref{prop:Fp-mon}). As for ordinary logic programs, the Fitting operator encapsulates the van Emden-Kowalski operator.

\begin{proposition}\label{prop:Fp-Tp} For every stream $\mb I$, $\Fp t(\mb I,\mb I)=\Tp t(\mb I)$.
\end{proposition}



\section{Fixed Point Semantics}\label{sec:semantics}

In this section, we define a fixed point semantics for the class of stream logic programs in terms of the Fitting operator defined above. More precisely, we first show that the Fitting operator is monotone with respect to the precision ordering, and conclude that certain least fixed points, the so-called $\Fp t$-answer streams, exist (cf. \yref{def:Fp-answer}). Then we compare our constructive semantics to the FLP-style semantics of \cite{Beck18} (cf. \yref{thm:main} and \yref{thm:level}).

\begin{proposition}\label{prop:Fp-mon} The Fitting operator $\Fp t$ is monotone.
\end{proposition}
\begin{proof} Let $\IJ$ and $(\mb I',\mb J')$ be 3-valued streams with $\IJ\seq_p (\mb I',\mb J')$. For an arbitrary rule $\rho\in\P$, $\IJ,t\modG\mm B(\rho)$ implies $(\mb I',\mb J'),t\modG\mm B(\rho)$ as a direct consequence of \yref{def:3-val}. Finally, \yref{prop:MMt-mon} implies $\Fp t\IJ\seq\Fp t(\mb I',\mb J')$.
\end{proof}

A consequence of \yref{prop:Fp-mon} is that in case $\mb I$ is a $t$-model of $\P$, $\Fp t(\,.\,,\mb I)$ is a monotone operator on the {\em complete} lattice $[\0,\mb I]$, since for every $\mb K\in [\0,\mb I]$, $$\Fp t(\mb K,\mb I)\seq\Fp t(\mb I,\mb I)=\Tp t(\mb I)\seq\mb I$$ holds by \yref{prop:Fp-Tp} and \yref{thm:Tp}. 

Define the operator $\Fp t\t$ on the set of all $t$-models of $\P$ by $$\Fp t\t(\mb I)=\lfp\Fp t(\,.\,,\mb I).$$ The soundness of $\Fp t\t$ is justified by the well-known Knaster-Tarski theorem which guarantees the existence of least fixed points of monotone operators on complete lattices.

We are now ready to formulate our fixed point semantics.

\begin{definition}\label{def:Fp-answer} We call every $t$-model $\mb I$ of $\P$ (for $\mb D$) a {\em $\Fp t$-answer stream} if $\mb I$ is a fixed point of $\Fp t\t$.
\end{definition}

For readers not familiar with the fixed point theory of logic programming, we briefly recall the basic intuitions behind \yref{def:Fp-answer} in the setting of ordinary logic programs. For the moment, let $\P$ be an ordinary program, and let $I$ be an interpretation of $\P$. The {\em Gelfond-Lifschitz reduct} of $\P$ with respect to $I$ is defined by{\footnote{Here $\mm B^-(\rho)$ denotes the negated atoms in the body of $\rho$, and $\mm B^+(\rho)$ denotes $\mm B(\rho)\sm\mm B^-(\rho)$.}} $$\P^I=\left\{\mm H(\rho)\la\mm B^+(\rho)\mid\rho\in\P:I\cap\mm B^-(\rho)=\0\right\}.$$ We call $I$ an {\em answer set} \citep{Gelfond91} of $\P$ if $I$ is the least model of $\P^I$, which coincides with $\lfp\mm T_{\P^I}$. So the computation of answer sets according to \cite{Gelfond91} is a two-step process, and \cite{Fitting02} showed how these two steps can be emulated by a single (monotone) operator, namely the Fitting operator $\Phi_\P$. Specifically, the identity $\Phi_\P(\,.\,,I)=\mm T_{\P^I}$ implies that $I$ is an answer set of $\P$ if and only if $I$ is the least fixed point of $\Phi_\P(\,.\,,I)$ or, equivalently, if $I$ is a fixed point of $\Phi_\P\t$. It is now clear that \yref{def:Fp-answer} is an extension of the ordinary answer set semantics to stream logic programs.

\begin{proposition} For an ordinary program $\P$, $I$ is an answer set of $\P$ if, and only if, $\mb I=I_t$ with $I_t=I$ is a $t$-answer stream of $\P$, for every $t\geq 1$.
\end{proposition}

We illustrate our fixed point semantics with the following example.

\begin{example}\label{exa:Fp} Reconsider the program $\P$ of \yref{exa:running} consisting of the following rules:
\begin{align*} 
@_2 a&\lra{1}\naf @_7 c       & \boxplus_{[1,\infty]}\Box c&\lra{3}\naf @_2 a\\
\boxplus_{[\infty,0]}\Box a&\lra{2}\naf c  & \boxplus_{[2,3]}\Box(a\land b)&\lra{4}\boxplus_{[0,1]}\Diamond c,\Box d.
\end{align*} We have argued in \yref{exa:running} that the interpretation stream $$\mb I=\{a\}_1\{a,b\}_3\{a,b,c\}_4\{a,b,c\}_5\{a,b,c\}_6\{a,b,c\}_7\{a,b,c\}_8\{c\}_9\{c\}_{10}$$ of $\P$ for $\mb D=\{a\}_1\{a,b\}_5\{c\}_{10}$ and $\Gamma=\{d\}$ is a 5-answer stream of $\P$. We now want to apply our tools from above to rigorously prove that $\mb I$ is a $\Fp 5$-answer stream by computing $\Fp 5\t(\mb I)$ bottom-up as follows. We start the computation with $\mb I_0=\0$:
\begin{align*} 
\Fp 5(\0,\mb I)= \mb D\cup\mm{MM}_{\0,5}(\mm H(\rho_3))
  = \mb D\cup\mm{MM}_{\0,5}(\boxplus_{[1,\infty]}\Box c)
  = \mb D\cup\mm{MM}_{\0[1,\infty;5],5}(\Box c)
  = \mb D
\end{align*} where the last equality follows from $\0[1,\infty;5]=\0$ and $\mm{MM}_{\0,5}(\Box c)=\0$. Then we continue the computation with $\mb I_1=\mb D$:
\begin{align*} 
\Fp 5(\mb I_1,\mb I)
  &= \mb D\cup\mm{MM}_{\mb I_1,5}(\mm H(\rho_3))\\
  &= \mb D\cup\mm{MM}_{\mb I_1,5}(\boxplus_{[1,\infty]}\Box c)\\
  &= \mb D\cup\mm{MM}_{\mb I_1[1,\infty;5],5}(\Box c)\\
  &= \mb D\cup\mm M_{\mm M_{\mb I_1[1,\infty;5],5}(\Box c),5}(\Box c)\\
  &= \mb D\cup\mm M_{\{c\}_4\ldots\{c\}_{10},5}(\Box c)\\
  &= \mb D\cup\{c\}_4\ldots\{c\}_{10}\\
  &= \{a\}_1\{c\}_4\{a,b,c\}_5\{c\}_6\{c\}_7\{c\}_8\{c\}_9\{c\}_{10}\\
  &= \mb I_2.
\end{align*} For the third iteration, we first compute
\begin{align*}
\mb M
  &=\mm M_{\mb I_2,5}(\{\mm H(\rho_3),\mm H(\rho_4)\})\\
  &=\mm M_{\mb I_2,5}(\boxplus_{[1,\infty]}\Box c\land\boxplus_{[2,3]}\Box(a\land b))\\
  &=\mm M_{\mb I_2,5}(\boxplus_{[1,\infty]}\Box c)\cup\mm M_{\mb I_2,5}(\boxplus_{[2,3]}\Box(a\land b))\\
  &=\mm M_{\mb I_2[1,\infty;5],5}(\Box c)\cup\mm M_{\mb I_2[2,3;5],5}(\Box(a\land b))\\
  &=\{a,b\}_3\{a,b,c\}_4\{a,b,c\}_5\{a,b,c\}_6\{a,b,c\}_7\{a,b,c\}_8\{c\}_9\{c\}_{10}
\end{align*} and then:
\begin{align*} 
\Fp 5(\mb I_2,\mb I)
  &= \mb D\cup\mm{MM}_{\mb I_2,5}(\{\mm H(\rho_3),\mm H(\rho_4)\})\\
  &= \mb D\cup\mm{MM}_{\mb I_2,5}(\boxplus_{[1,\infty]}\Box c\land\boxplus_{[2,3]}\Box(a\land b))\\
  &= \mb D\cup\mm M_{\mb M,5}(\boxplus_{[1,\infty]}\Box c)\cup\mm M_{\mb M,5}(\boxplus_{[2,3]}\Box(a\land b))\\
  &= \mb D\cup\mm M_{\mb M[1,\infty;5]}(\Box c)\cup\mm M_{\mb M[2,3;5],5}(\Box(a\land b))\\
  &= \mb D\cup\mb M\\
  &= \{a\}_1\{a,b\}_3\{a,b,c\}_4\{a,b,c\}_5\{a,b,c\}_6\{a,b,c\}_7\{a,b,c\}_8\{c\}_9\{c\}_{10}\\
  &= \mb I_3.
\end{align*} Finally, we verify that $\mb I=\mb I_3$ is a fixed point of (cf. \yref{exa:Tp} and \yref{prop:Fp-Tp}):
\begin{align*}
\Fp 5(\mb I,\mb I)=\Tp 5(\mb I)=\mb I.
\end{align*} In summary, the above computations show that $\mb I$ is a fixed point of $\Fp 5\t$ or, equivalently, a $\Fp 5$-answer stream.
\end{example}

We now wish to relate our fixed point semantics from above to the FLP-style semantics of \cite{Beck18} presented in \yref{sec:LARS}. Firstly, we prove some auxiliary lemmata.

\begin{lemma}\label{lem:3-val} Let $\IJ$ be a 3-valued stream, and let $\mb K\in[\mb I,\mb J]$. Then, $\Phi_{\P^{\mb K,t},\mb D,t}\IJ=\Fp t\IJ$.
\end{lemma}
\begin{proof} Define $\P^{\IJ,t}=\{\rho\in\P\mid \IJ,t\modG\mm B(\rho)\}.$ As a direct consequence of 3-valued entailment (cf. \yref{def:3-val}), we have the following inclusions:
\begin{align}\label{equ:inclusions} \P^{\IJ,t}\seq\P^{\mb K,t}\seq\P.
\end{align} By the monotonicity of $\MMt$ (cf. \yref{prop:MMt-mon}), $\Phi_{\mr R,\mb D,t}\IJ\seq\Fp t\IJ$ holds whenever $\R\seq\P$, for all programs $\P$ and $\R$. Therefore, we can conclude from \yref{equ:inclusions}:
\begin{align}\label{equ:inclusions2} \Phi_{\P^{\IJ,t},\mb D,t}\IJ\seq\Phi_{\P^{\mb K,t},\mb D,t}\IJ\seq\Fp t\IJ.
\end{align} Note that by definition, we have $\Phi_{\P^{\IJ,t},\mb D,t}\IJ=\Fp t\IJ$ which together with \yref{equ:inclusions2} entails $\Phi_{\P^{\mb K,t},\mb D,t}\IJ=\Fp t\IJ$.
\end{proof}

\begin{lemma}\label{lem:lower} For every prefixed point $\mb K\seq\mb I$ of $\Tp t$, $\Fp t\t(\mb I)\seq\mb K$.
\end{lemma}
\begin{proof} We compute $\Fp t\t(\mb I)$ bottom-up. Clearly, $\mb K_0=\0\seq\mb I$. Since $\Fp t(\,.\,,\mb I)$ is monotone, we have
\begin{align*} \mb K_1=\Fp t(\0,\mb I)\seq\Fp t(\mb K,\mb K)=\Tp t(\mb K)\seq\mb K.
\end{align*} Similarly, we can compute $\mb K_2=\Fp t(\mb K_1,\mb I)\seq\mb K$ and so on, which shows that the limit $\Fp t\t(\mb I)$ is contained in $\mb K$, i.e., $\Fp t\t(\mb I)\seq\mb K$.
\end{proof}

We are now ready to prove the main result of this paper.

\begin{theorem}\label{thm:main} Every $\Fp t$-answer stream is a $t$-answer stream of $\P$.
\end{theorem}
\begin{proof} By assumption, we have $\Fp t\t(\mb I)=\mb I$ which implies 
\begin{align}\label{equ:main}\Phi_{\P\It,\mb D,t}\t(\mb I)=\mb I
\end{align} by \yref{lem:3-val}, that is, $\mb I$ is a $\Phi_{\P\It,\mb D,t}$-answer stream. Since every $\Phi_{\P\It,\mb D,t}$-answer stream is a $t$-model of $\P\It$, it remains to show that $\mb I$ is a {\em minimal} $t$-model of $\P\It$. For this suppose there exists some stream $\mb K$ with $\mb K\sneq\mb I$ such that $\mb K$ is a $t$-model of $\P\It$. Then, by \yref{thm:Tp}, we have $\mm T_{\P\It,\mb D,t}(\mb K)\seq\mb K\sneq\mb I$ which implies $$\Phi_{\P\It,\mb D,t}\t(\mb I)\seq\mb K\sneq\mb I$$ by \yref{lem:lower}---a contradiction to \yref{equ:main}.
\end{proof}

\yref{thm:main} shows that our fixed point semantics is sound with respect to our FLP-style semantics. However, the next example shows that the converse of \yref{thm:main} fails in general.

\begin{example}\label{exa:R2} Reconsider the program $\R$ of \yref{exa:R} consisting of the rules
\begin{align*} 
a&\lra{1}\Box b\\
b&\lra{2}\Box a.
\end{align*} In \yref{exa:R} we have seen that $\{a,b\}_t$ is a $t$-answer stream of $\R$, for every $t\geq 1$ (and $\mb D=\Gamma=\0$). On the other hand, we have $\Phi_{\R,t}(\0,\{a,b\}_t)=\0$ which shows that $\{a,b\}_t$ is {\em not} a $\Phi_{\R,t}$-answer stream.
\end{example}

\section{Level Mappings}\label{sec:level}

In this section, we define level mappings for stream logic programs in the vein of \cite{Shen14}, and prove in \yref{thm:level} that $\Fp t$-answer streams characterize those $t$-models which posses a level mapping or, equivalently, which are free of circular justifications.

Firstly, we recast the notion of a partitioning to stream logic programs.

\begin{definition} A {\em partitioning} of a stream $\mb I$ is a sequence of streams $\mb S=(\mb S_0,\mb S_1,\ldots,\mb S_m)$, $m\geq 1$, where $\mb S_0=\0$, $\mb S_1\cups\mb S_m=\mb I$, $\mb S_i\neq\0$ for every $i\geq 1$, and $\mb S_i\cap\mb S_j=\0$ for every $i\neq j\neq 0$.
\end{definition}

We now define level mappings over such partitionings.

\begin{definition}\label{def:level} A {\em $t$-level mapping} of a stream $\mb I$ with respect to $\P$ (for $\mb D$) is a partitioning $\mb S=(\mb S_0,\mb S_1,\ldots,\mb S_m)$ of $\mb I$ such that for all $1\leq i\leq m$,
\begin{align}\label{equ:S} \mb S_i\seq\mb D\cup\mm{MM}_{\mb S_1\cups\mb S_{i-1},t}(\{\mm H(\rho)\mid\rho\in\P:(\mb S_1\cups\mb S_{i-1},\mb I),t\modG\mm B(\rho)\}).
\end{align} We call $\mb S$ a {\em total} $t$-level mapping of $\mb I$ if in addition $\mb I=\mb S_0\cups\mb S_m$ is a $t$-model of $\P$.
\end{definition}

The intuition behind \yref{def:level} is as follows. A partitioning $\mb S=(\mb S_0,\mb S_1,\ldots,\mb S_m)$ with $\mb S_i=S_{i,1}S_{i,2}\ldots$, $1\leq i\leq m$, is a $t$-level mapping of $\mb I$ if each proposition $a\in S_{i,t_i}$ (i.e., $a$ holds in level $i$ at time point $t_i$) is non-circularly justified by the rules in $\P$, i.e., either $a\in D_{t_i}$ or there exists a rule $\rho$ in $\P$ justifying $a$ at time point $t_i$, that is, $a$ occurs in the head of $\rho$ and the body of $\rho$ ``fires'' in a level smaller than $i$. For $\mb S$ to be called total, we additionally require $\mb S_1\cups\mb S_m=\mb I$ to contain every proposition occurring in a rule head which is derivable from $\mb I$, i.e., $\mb D\cup\MMt(\mm H(\P\It))\seq\mb I$ which, by \yref{thm:Tp}, is equivalent to $\mb I$ being a $t$-model of $\P$. Clearly, a stream $\mb I$ possessing a $t$-level mapping is free of circular justifications.

Note that we can rewrite \yref{equ:S} more compactly as
\begin{align}\label{equ:Fp-S} \mb S_i\seq\Fp t(\mb S_1\cups\mb S_{i-1},\mb I)
\end{align} which shows the direct relationship between $t$-level mappings and the Fitting operator.

\begin{example}\label{exa:level} Once again, reconsider the program $\P$ of \yref{exa:running}. In \yref{exa:Fp} we have seen that $$\mb I=\{a\}_1\{a,b\}_3\{a,b,c\}_4\{a,b,c\}_5\{a,b,c\}_6\{a,b,c\}_7\{a,b,c\}_8\{c\}_9\{c\}_{10}$$ is a $\Fp 5$-answer stream for  $\mb D=\{a\}_1\{a,b\}_5\{c\}_{10}$. We construct the total 5-level mapping $\mb S=(\mb S_0,\mb S_1,\mb S_2,\mb S_3)$ of $\mb I$ for $\P$ as follows:{\footnote{By ``$\sm$'' we mean here the point-wise relative complement, e.g., $\{a\}_1\{b\}_2\sm\{b\}_2=\{a\}_1$.}}
\begin{align*} 
\mb S_0 &= \mb I_0 = \0\\
\mb S_1 &= \mb I_1\sm\mb I_0 = \mb D = \{a\}_1\{a,b\}_5\{c\}_{10}\\
\mb S_2 &= \mb I_2\sm\mb I_1 = \{c\}_4\{c\}_5\{c\}_6\{c\}_7\{c\}_8\{c\}_9\\
\mb S_3 &= \mb I_3\sm\mb I_2 = \{a,b\}_3\{a,b\}_4\{a,b\}_6\{a,b\}_7\{a,b\}_8
\end{align*} where $\mb I_0=\0,\mb I_1=\mb D,\mb I_2$, and $\mb I_3=\mb I$ are the intermediate results in the bottom-up computation of $\Fp 5\t(\mb I)$ (cf. \yref{exa:Fp}).
\end{example}

We can characterize the $\Fp t$-answer streams in terms of $t$-level mappings as follows.

\begin{theorem}\label{thm:level} A stream $\mb I$ is a $\Fp t$-answer stream if, and only if, there is a total $t$-level mapping $\mb S$ of $\mb I$ with respect to $\P$.
\end{theorem}
\begin{proof} For the direction from left to right, we construct the total $t$-level mapping $\mb S$ of the $\Fp t$-answer stream $\mb I$ as in \yref{exa:level}. Let $\mb I_0=\0,\mb I_1,\ldots,\mb I_m=\mb I$ be the intermediate results of the bottom-up computation of $\Fp t\t(\mb I)=\mb I$, i.e., $$\Fp t(\mb I_{i-1},\mb I)=\mb I_i,\quad 1\leq i\leq m,$$ and define $\mb S_0=\0$ and $\mb S_i=\mb I_i\sm\mb I_{i-1}$, for all $1\leq i\leq m$. By construction, we have $\mb I_i=\mb S_1\cups\mb S_i$, for all $1\leq i\leq m$, which directly yields the inclusion in \yref{equ:Fp-S}; moreover, since $\mb I$ is a $t$-model of $\P$, $\mb S$ is a total $t$-level mapping of $\mb I$ with respect to $\P$.

For the opposite direction, let $\mb S=(\mb S_0,\mb S_1,\ldots,\mb S_m)$, $m\geq 1$, be a total $t$-level mapping of $\mb I$ with respect to $\P$. We need to show that $\mb I=\bigcup\mb S$, with $\bigcup\mb S=\mb S_1\cups\mb S_m$, is a fixed point of $\Fp t\t$. Since $\mb S$ is total, $\mb I$ is a $t$-model of $\P$, so we have by \yref{equ:Fp-S}, the monotonicity of $\Fp t$ (cf. \yref{prop:Fp-mon}), \yref{prop:Fp-Tp}, and \yref{thm:Tp}:
\begin{align*} 
\mb I=\bigcup\mb S\seq\Fp t(\mb S_1\cups\mb S_{m-1},\mb I)\seq\Fp t\left(\bigcup\mb S,\mb I\right)=\Fp t(\mb I,\mb I)=\Tp t(\mb I)\seq\mb I.
\end{align*} So $\mb I$ is a fixed point of $\Fp t(\,.\,,\mb I)$ and it remains to show that there is no fixed point $\mb K\sneq\mb I$ of $\Fp t(\,.\,,\mb I)$. Suppose, towards a contradiction, that for some $\mb K\sneq\mb I$,
\begin{align} \Fp t(\mb K,\mb I)=\mb K.
\end{align} Since $\mb K\sneq\mb I$ there is some $i$, $1\leq i\leq m$, such that $\mb S_1\cups\mb S_{i-1}\seq\mb K\seq\mb S_1\cups\mb S_i$. So by \yref{equ:Fp-S} and \yref{prop:Fp-mon} we have
\begin{align*} \mb S_i\seq\Fp t(\mb S_1\cups\mb S_{i-1},\mb I)\seq\Fp t(\mb K,\mb I)=\mb K.
\end{align*} Consequently,
\begin{align*}
\mb K
  &=\Fp t(\mb K,\mb I)\\
  &\seq\Fp t(\mb S_1\cups\mb S_i,\mb I)\\
  &\seq\Fp t(\mb S_1\cups\mb S_{i-1}\cup\mb K,\mb I)\\
  &=\Fp t(\mb K,\mb I)\\
  &=\mb K,
\end{align*} which implies
\begin{align}\label{equ:Sb_i} \Fp t(\mb S_1\cups\mb S_i,\mb I)=\mb K.
\end{align} From \yref{equ:Fp-S}, \yref{equ:Sb_i}, and the monotonicity of $\Fp t$ (cf. \yref{prop:Fp-mon}) we infer
\begin{align*}
\mb S_{i-1}&\seq\Fp t(\mb S_1\cups\mb S_{i-2},\mb I)\seq\Fp t(\mb S_1\cups\mb S_i,\mb I)=\mb K;\\
\mb S_{i+1}&\seq\Fp t(\mb S_1\cups\mb S_i,\mb I)=\mb K.
\end{align*} Hence, $\mb S_j\seq\mb K$ for all $1\leq j\leq m$, and so $\bigcup\mb S\seq\mb K\sneq\mb I$---a contradiction to $\bigcup\mb S=\mb I$.
\end{proof}

\begin{example}\label{exa:level2} Reconsider the program $\R$ of \yref{exa:R} consisting of the following two rules:
\begin{align*} 
a&\la\Box b\\
b&\la\Box a.
\end{align*} In \yref{exa:R} we have seen that for every $t\geq 1$, $\mb I=\{a,b\}_t$ is a $t$-answer stream of $\R$. Note that $a$ and $b$ are circularly justified in $\R$. As $\mb I$ is {\em not} a $\Fp t$-answer stream (cf. \yref{exa:R2}), there is no total $t$-level mapping of $\mb I$ by \yref{thm:level}. 
\end{example}


Note that \yref{thm:level} together with \yref{thm:main} (and \yref{exa:R2}) characterize our semantics as the strict constructive subclass of our FLP-style semantics.

\section{Conclusion}\label{sec:Conc}

This paper contributed to the foundations of stream reasoning \citep{Valle09,Mileo17,DellAglio17} by providing a sound and constructive extension of the answer set semantics from ordinary to stream logic programs. For this we refined the FLP-style semantics of \cite{Beck18}. Moreover, we extended the van Emden-Kowalski and Fitting operators from ordinary to stream logic programs. As a result of our investigations, we obtained constructive semantics of stream logic programs with nice properties. More precisely, it turned out that our fixed point semantics can be characterized in terms of level mappings or, equivalently, is free of circular justifications, which is regarded as a positive feature. Moreover, the algebraic nature of our fixed point semantics yields computational proofs which are satisfactory from a mathematical point of view.

As our fixed point semantics hinges on the (extended) Fitting operator, it can be reformulated within the algebraic framework of Approximation Fixed Point Theory (AFT) \citep{Denecker04,Denecker12}, which is grounded in the work of Fitting on bilattices in logic programming (cf. \cite{Fitting02}), and which captures a number of related (non-monotonic) formalisms (e.g., \cite{Denecker03,Antic13}). In the future, we wish to apply the full framework of AFT to LARS, which provides a well-founded semantics \citep{vanGelder91}, a notion of strong and uniform equivalence \citep{Truszczynski06}, a bottom-up semantics for disjunctive programs \citep{Antic13}, and a recently introduced algebraic notion of groundedness \citep{Bogaerts15}. 


\section*{Acknowledgments}

We would like to thank the reviewers for their thoughtful and constructive comments, and for their helpful suggestions to improve the presentation of the article.

\bibliographystyle{elsarticle-harv}\biboptions{authoryear}
\bibliography{2020-01-08-fixstream}
\end{document}